\documentclass[letterpaper, 10 pt, conference]{ieeeconf}  
\let\proof\relax 
\let\endproof\relax
\usepackage{amsthm,amsmath,amssymb,mathtools,newtxmath}
\usepackage{graphicx}
\usepackage{cite}
\usepackage{cleveref}

\newtheorem{Proposition}{Proposition}

\newtheorem{theorem}{Theorem}
\newtheorem{remark}{Remark}
\newtheorem{definition}{Definition}
\newtheorem{corollary}{Corollary}
\usepackage{algorithm}
\usepackage{algorithmic}
\usepackage{booktabs,lipsum}
\usepackage{color}
\usepackage{multirow}

\DeclareMathOperator*{\argmax}{arg\,max}
\DeclareMathOperator{\Tr}{Tr}
\DeclareMathOperator{\Id}{Id}

\IEEEoverridecommandlockouts                              

\title{\LARGE \bf
Cross-Layer Coordinated Attacks on Cyber-Physical Systems: A LQG Game Framework with Controlled Observations
}

\author{Yunhan Huang$^{1}$, Zehui Xiong$^{2}$ and Quanyan Zhu$^{1}$
\thanks{$^{1}$ Y. Huang and Q. Zhu are with the Department of Electrical and Computer Engineering,
        New York University, 370 Jay St., Brooklyn, NY.
        {\tt\small \{yh.huang, qz494\}@nyu.edu}}
\thanks{$^{2}$ Z. Xiong is with the Pillar of Information Systems Technology and Design, Singapore University of Technology and Design, Singapore.
        {\tt\small zehui.xiong@ieee.org}     }
}

\begin{document}

\maketitle
\thispagestyle{plain}
\pagestyle{plain}

\begin{abstract}

This work establishes a game-theoretic framework to study cross-layer coordinated attacks on cyber-physical systems (CPSs). The attacker can interfere with the physical process and launch jamming attacks on the communication channels simultaneously. At the same time, the defender can dodge the jamming by dispensing with observations. The generic framework captures a wide variety of classic attack models on CPSs. Leveraging dynamic programming techniques, we fully characterize the Subgame Perfect Equilibrium (SPE) control strategies. We also derive the SPE observation and jamming strategies and provide efficient computational methods to compute them. The results demonstrate that the physical and cyber attacks are coordinated and depend on each other.

On the one hand, the control strategies are linear in the state estimate, and the estimate error caused by jamming attacks will induce performance degradation. On the other hand, the interactions between the attacker and the defender in the physical layer significantly impact the observation and jamming strategies.  Numerical examples illustrate the exciting interactions between the defender and the attacker through their observation and jamming strategies.
\end{abstract}

\section{Introduction}

Recent progress in information and communications technologies (ICT) such as the Internet of Things (IoT) and 5G high-speed cellular networks have enhanced the connectivity among physical systems and cyber systems. However, the increasing connectivity also brings with these systems heightened concern about trustworthiness. There is an urgent need for understanding security, privacy, safety, reliability, resilience, and corresponding assurance for CPSs. Due to the multi-layer and multi-stage nature of CPSs, a cross-layer cross-stage framework is a sine qua non to understand the trustworthiness of CPSs. Most existing works on the security of CPSs often focus independently on either the physical system or the cyber system. One common assumption is that adversaries can only launch one particular type of attack at a time. For example, in \cite{feng2020networked,foroush2012event, gupta2010optimal, zhang2015optimal,xu2017game}, the authors have considered DoS attacks that jam either the observation or the control signals to deteriorate the performance of the underlying system. \cite{mo2013detecting} focuses only on data injection attacks on the sensors of a control system. Studies in \cite{zhu2013performance} pivot purely on the replay attacks on the operator-to-actuator channel. However, in CPSs, adversaries can leverage both cyber and physical vulnerabilities to launch coordinated attacks \cite{zhu2018multi,rass2017physical}. For example, an advanced adversary can simultaneously compromise critical sensors and control units to damage targeted CPS assets.

In this work, we build a dynamic game-theoretic framework that can incorporate various attack models, where the attacker conducts both physical interferences (e.g., through either direct physical intervention or cyber hacking/data injection) and jamming attacks on observation channel. The attacker has to intelligently coordinate her/his attacks across both the physical and the cyber layers over a finite period of time to maximize the system degradation with minimum effort.  The defender, e.g., the controller/operator, implements her/his own control and has to, at each time, decide whether to observe or not. Each observation request made by the defender is associated with an observation cost. The cost can capture the limited network resources, such as power, communication, and bandwidth \cite{huang2019continuous}. For example, radar measurement requires megawatts of power for sensing in military applications. Thus, the purposes of the defender not observing are two fold: One is to save limited resources; the other is to dodge the jamming. The physical process is described by a linear dynamical system with additive white noise. The observation is partial and noisy, whose availability depends on the observation decisions of the defender and the jamming policies of the attacker.

Dynamic games have long been used to capture the cross-layer multi-stage nature of CPSs and the competing nature between the system operator and the adversary \cite{yunhanhuang2020dynamic}. At the cyber layer of the CPSs, the dynamic games are used to model the cyber kill chains of APTs that include reconnaissance, lateral movement, and command and control. These games are often built over graphical models such as computer networks \cite{noureddine2016game,horak2017manipulating} and attack graphs \cite{nguyen2017multi}. At the physical layer, the dynamic games are used to describe the interactions between operational technologies (OT) and an adversary. The game-theoretic description of the threat model at the OT level guides the design of security monitoring, and control strategies that aim to reduce the risks on the controlled processes and assets \cite{zhu2018multi}, as well as the development of resilient control mechanisms that mitigate the impact of successful attacks \cite{chen2019control}.

Our work's main contribution is the development of a cross-layer over-stage game-theoretic framework that underpins the study of CPS under coordinated simultaneous cross-layer attacks. The generic framework captures various attack models. Its connections with existing attack models will be discussed in \Cref{remark:versatility}. 

We study the SPE strategies of this dynamic game and fully characterize the SPE strategies via two dynamic programming equations. The theoretical results show that the control strategies are linear in the state estimate. The control gain can be computed offline, independent of the observation and the jamming strategies. The effect of jamming attacks causes the estimation error and leads to significant system degradation. The capability of physical attacks also affects the jamming and observation decisions. The SPE strategies show that the defender does not observe when the observation cost surpasses system degradation caused by estimation error.  When the attacker can deploy physical attacks with low costs, there is no incentive for the defender to observe even when the observation cost is zero. Otherwise, the attacker can leverage the observation to launch more accurate physical attacks. Besides, we show that the defender makes an observation even when the defender anticipates the jamming to incur a higher cost to the attacker. More results regarding the jamming and the observation decisions will be discussed in \Cref{Subsec:ObservationStrategies} and \Cref{sec:NumericalStudy}. 

For the rest of the paper, \Cref{sec:ProblemForm} provides a mathematical formulation of the problem. \Cref{Sec:TheoreticalResults} presents the theoretical results of the framework, their derivations, and computational methods. In \Cref{sec:NumericalStudy}, we use case studies to illustrate the structures of the observation and jamming strategies.

\section{Problem Formulation}\label{sec:ProblemForm}

We consider a linear dynamical system given by
\begin{equation}\label{Eq:SystemDynamics}
x_{n+1}  = A x_n + B^d u_n^d + B^a u^a_n + C w_n,\ \ \ 0\leq n\leq N-1,
\end{equation}
where $x_n \in\mathbb{R}^q$ is the $q$-dimensional state vector at time $n$; and $u^d_n, u_n^a, w_n$ are vectors of dimensions less than or equal to $q$. Here, $u^d_n$ is the control of the defender while $u^a_n$ is that of the attacker. The system noise at time $n$ is denoted by $w_n$. Moreover, $A$, $B^d$, $B^a$, and $C$ are matrices with appropriate dimensions. The capability of the attacker on affecting the state dynamics can be captured by the adversarial control matrix $B^a$. Here, we consider time-invariant system for notational simplicity, but all results in \Cref{Sec:TheoreticalResults} can be extended to the time-varying case without much endeavor.
The associated observation system is 
\begin{equation}\label{Eq:ObservationSystem}
\begin{aligned}
\tilde{y}_n &= D x_n + E v_n,\\
y^d_n &= h^d(i^d_n,i^a_n) \cdot \tilde{y}_n,\ \textrm{and }y^a_n = h^a(i^d_n,i^a_n) \cdot \tilde{y}_n,
\end{aligned}
\end{equation}
where $v$ and $D,E$ are vectors and matrices with appropriate dimensions. Here, $i_n^d\in \{0,1\}$ (resp. $i^a_n\in \{0,1\}$) is the observation (resp. jamming) decision made by the defender (resp. the attacker). We call $\tilde{y}_n$ the information vector at time $n$. Whether or not the vector $\tilde{y}_n$ is observed by the defender and the attacker is decided by the observation decision $i^d_n$ and the jamming decision $i^a_n$ according to the rule $h:\{0,1\} \times \{0,1\} \rightarrow \{0,1\}$. We suppose, whenever the attacker chooses to jam the observation, the defender will receive no information. The attacker receives the same observation information as the defender does. In this case, $h^d(i^d_n,i^a_n) = h^a(i^d_n,i^a_n) = i^d_n \cdot (1-i_{n}^a)$. Hence, $y^d_n = y^a_n$ and so we use $y_n$ instead in later discussions.

We introduce the notation $X_n = \{x_0,\cdots, x_n\}$ to denote the history of state trajectory up to time $n$. Similarly, we define $U_n^d$, $U^a_n$, $I_n^d$, $I_n^a$, $W_n$, $V_n$, $Y_n^d$, and $Y_n^a$ for $u_n^d$, $u^a_n$, $i^d_n$, $i^a_n$, $w_n$, $v_n$, $y^d_n$, and $y^a_n$, respectively. The sequences $W_{N-1}$ and $V_{N-1}$ are independent stochastic processes with a joint Gaussian probability distribution described by $\mathbb{E}[w_n] = \mathbb{E}[v_n] =0$, $\mathbb{E}[w_nw_{n'}] = \Sigma_s \delta(n-n')$, and $\mathbb{E}[v_nv_{n'}] = \Sigma_o \delta(n-n')$, where $\delta(\cdot)$ is the Kronecker delta. The initial condition $x_0$, independent from the system noise and the observation noise, is Gaussian distributed with mean $\bar{x}_0$ and variance $\Sigma_0$. 

\textbf{\textit{Information:}} The control sequences $U^d_{N-1}$ (resp. $U^a_{N-1}$) and the observation/jamming sequence $I_{N-1}^d$ (resp. $I_{N-1}^a$) are to be generated by the defender (resp. attacker). At time $n$, the observation $i^d_n$ is made based on the information available to the defender, which is denoted by
\begin{equation}\label{Eq:InformationBeforeObservation}
\mathcal{F}_n = \{I_{n-1}^d,I_{n-1}^a,U^d_{n-1},U_{n-1}^a, Y_{n-1}\},\ \ \ n\geq 1,
\end{equation}
with $\mathcal{F}_0 = \varnothing$. So is the jamming decision $i^a_n$. The controls of both defender and attacker are made based on the information available at time $n$ after the observation, which is denoted by
\begin{equation}\label{Eq:InformationAfterObservation}
\bar{\mathcal{F}}_n = \{\mathcal{F}_n, I_n^d, I_n^a, Y_n\}.
\end{equation}

We assume that the state evolution equations, observation equations, noise statistics, cost functions of the controllers, and information structures of the controllers are part of common knowledge among the players. The game is hence a complete information game.

\textbf{\textit{Objectives/Targets:}} The cost functional of the defender and the attacker, involving quadratic costs in state and their controls, as well as the observation and  cost, can be written as
\begin{equation}\label{Eq:CostFunctional}
\begin{aligned}
F^d(\pi^d,\pi^a) = \mathbb{E}&\left[\sum_{n=0}^{N-1} c_n(x_n,u^d_n,u^a_n,i^d_n,i^a_n)  + x_N' Q_N x_N  \right],\\
\end{aligned}
\end{equation}
where 
$$
\begin{aligned}
c_n(x_n,u^d_n,u^a_n,i^d_n,i^a_n) =& x_n' Q_n x_n + {u_n^d}'R_n^d u_n^d - {u^a_n}' R_n^a u_n^a \\
&+ i^d_n O^d_n - i^a_n O^a_n,
\end{aligned}
$$
is the instantaneous cost at stage $n$, and we have $F^a(\pi^d,\pi^a) = - F^d(\pi^d,\pi^a)$. Here, the matrices $R_n^d$, $R_n^a$ are positive definite, the matrix $Q_n^d$ is positive semidefinite, and the scalars $O_n^d$, $O_n^a$ are nonnegative for $n=1,2,\cdots,N-1$. Here, $O_n^d$ represents the observation cost for the defender while $O_n^a$ denotes the jamming cost. For any matrix $M$, $M'$ indicates the transpose of $M$. 

\textbf{\textit{Strategies:}} $\pi^d = (\mu^d,\nu^d)$ is the strategies of the defender, where $\mu^d$ denotes the observation strategy and $\nu^d$ denotes the control strategy. The strategy of the attacker, including the jamming strategy $\mu^a$ and the control strategy $\nu^a$, are denoted by $\pi^a = (\mu^a,\nu^a)$. Given $\pi^d$, at stage $n$, the control and the observation decisions of the defender are generated as $i^d_n = \mu_n^d(\mathcal{F}_{n})$ and $u^d_n = \nu_n^d(\bar{\mathcal{F}_n})$.

The defender aims to stabilize the system with minimum control effort and at the same time observe/sample economically. The attacker possesses an opposing objective, which is to undermine the defender's effort by cross-layer coordinated attacks on both the physical layer and the communication layer. Here, we consider a zero-sum game where the defender and the attacker are strictly competitive. Our results in \Cref{Sec:TheoreticalResults} can be easily extended to a general sum setting.

\begin{remark}\label{remark:versatility}
The framework can also capture various attack models. For example, $R_n^a$ going to infinity means zero physical attacks. Hence, the framework specializes to optimal jamming attacks studied in \cite{zhang2015optimal,gupta2010optimal}; Letting $B^a = B^d$, we can model false data injection attacks in the operator-to-actuator channel \cite{sargolzaei2019detection}; With $h(i_d,i_a) = 1-(1-i_d)(1-i_a)$, the framework describes pursuit-evasion type of security problems with controlled information \cite{huang2021pursuit}, where detecting your opponent's location will expose your own location.
\end{remark}

\section{THEORETICAL RESULTS}\label{Sec:TheoreticalResults}

The non-hierarchical decision making between the defender and the attacker makes Nash equilibrium a natural solution concept for this game. In a Nash equilibrium, none of the players can get better off by unilaterally deviating from the equilibrium.

\begin{definition}
A strategy pair $({\pi^d}^*,{\pi^a}^*)$ is called a Nash equilibrium of the zero-sum dynamic game described by \Cref{Eq:SystemDynamics,Eq:ObservationSystem,Eq:InformationBeforeObservation,Eq:InformationAfterObservation,Eq:CostFunctional} if 
$$
F^d({\pi^d}^*,{\pi^a}) \leq F^d({\pi^d}^*,{\pi^a}^*) \leq F^d({\pi^d},{\pi^a}^*),
$$
for all $\pi^d \in \Pi^d$ and $\pi^a\in\Pi^a$, where $\Pi^d$($\Pi^a$) is the strategy space of the defender(attacker) that contains all possible strategies.
\end{definition}

The characterization and the computation of Nash equilibrium involve solutions of individual optimization problem for both the defender and the attacker. In a finite-horizon dynamic game, the characterization and the computation of the equilibrium are usually obtained by conducting backward induction, which gives rise to the concept of Subgame Perfect Nash Equilibrium (SPNE). The SPNE is a refinement of Nash equilibrium used in dynamic games with perfect information. A game with perfect information is a game where each player is perfectly informed of the history of what has happened so far, up to the point where it is her turn to move. Hence, this game is a game with perfect information. The expected cost-to-go of the defender conditioned on the information set from the beginning of time $k$ to its terminal state at time $N$ is
\begin{equation}\label{Eq:ExpectedCosttoGo}
\begin{aligned}
f_k^d(\mathcal{F}_k) = \mathbb{E}&\left[\sum_{n=k}^{N-1}c_n(x_n,u^d_n,u^a_n,i^d_n,i^a_n) + x_N' Q_N x_N \middle\vert \mathcal{F}_k  \right],
\end{aligned}
\end{equation}
for $k = 0,1,\cdots,N$. The expected cost-to-go functional of the attacker is hence $f_k^a(\mathcal{F}_k) =  -f^d_k(\mathcal{F}_k)$. For each stage $k$, the defender and the attacker, their cost-to-go functional $f^d_k$ and $f^a_k$, together with the strategies for future stages $(\pi^d_k, \pi^d_{k+1},\cdots,\pi^d_{N-1})$ and $(\pi^a_k, \pi^a_{k+1},\cdots,\pi^a_{N-1})$ constitute a subgame embedded in the original game in the original game. The original game is a subgame of itself when $k=0$.

\begin{definition}
An  equilibrium  is  an  SPNE  if  and  only  if  it  is  a  Nash  equilibrium in  every  subgame  of  the  original game.
\end{definition}
An SPNE is a Nash equilibrium for the entire game since the entire game of also a subgame when $k=0$. In this paper, we focus on studying the SPNE of the game. The complete characterization and the computation of the SPNE are conducted via two steps. The first is to characterize the SPNE control strategies from backward induction for all possible observation decision sequences. The second is to find the SPNE observation strategies based on the values under the SPNE control strategies computed in the first step.

\subsection{Control Strategies}\label{subsec:SPNEControlStrategies}

Suppose that we are given a sequence of observation decisions $I_{N-1}^d$ and a sequence of jamming decisions $I_{N-1}^a$. Under control strategies $\nu^d$ and $\nu^a$, the expected cost-to-go starting from time $k$ conditioning on the information available after the observation is
\begin{equation}\label{Eq:DefinitionPostObservationCost2Go}
\begin{aligned}
V_k^d(\nu^d,\nu^a) = \mathbb{E}&\left[\sum_{n=k}^{N-1}c_n(x_n,u^d_n,u^a_n,i^d_n,i^a_n) + x_N' Q_N x_N \middle\vert \bar{\mathcal{F}}_k  \right].
\end{aligned}
\end{equation}
Let us define the SPNE cost-to-go value as ${V_k^d}^*(\bar{\mathcal{F}}_k) = \min_{v^d} V_k^d(\nu^d,{\nu^a}^*)$, where ${\nu^a} = \argmax V^d_k({\nu^d}^*,\nu^a)$. The complete solution of this problem requires the knowledge of 1) the SPNE control strategies at any stage 2) the SPNE expected cost of proceeding from any state at any time to the end.  The main results for the SPNE control strategies are summarized in the following theorem.

\begin{theorem}
For any given observation sequence $I_{N-1}^d$ and jamming sequence $I_{N-1}^a$, starting from any stage $k=0,1,\cdots,N$, the SPNE cost-to-go value to the end is
\begin{equation}\label{Eq:SPNEControlCost2Go}
\begin{aligned}
{V_k^d}^* =& \mathbb{E}\left[x_k' L_{k} x_k\middle\vert \bar{\mathcal{F}}_k\right] + \sum_{n=k}^{N-1} \Tr \Sigma_s C'L_{n+1}C\\
&+   \sum_{n=k}^{N-1}\left( \Tr P_n(\bar{\mathcal{F}}_n) \varphi_n + i_n^d O_n^d - i_n^a O_n^a\right),
\end{aligned}
\end{equation}
where
{\small
\begin{equation}\label{Eq:RicattiUpdate}
\begin{aligned}
\end{aligned}
L_n = Q_n + A'\left(L_{n+1} -L_{n+1}\begin{bmatrix}
B^d & B^a\\
\end{bmatrix} M_{n}^{-1}\begin{bmatrix}
{B^d}'\\ {B^a}'
\end{bmatrix}L_{n+1} \right)A,
\end{equation}}for $n=1,2\cdots,N-1$ with $L_N = Q_N$. The matrix 
\begin{equation}\label{Eq:SPNEControlAuxiliaryMatrix}
M_n =  \begin{bmatrix}
R_{n}^d + {B^d}'L_{n+1}B^d & {B^d}'L_{n+1}B^a\\
{B^a}' L_{n+1} B^d & {B^a}' L_{n+1} B^a - R^a_{n}
\end{bmatrix}
\end{equation}
is invertible provided that $R_n^a > {B^a}' L_{n+1} B^a$ for $k=0,1,\cdots,N-1$. 

At each stage $n$, the SPNE control strategies of the defender and the attacker take the form of the linear state feedback control laws
\begin{equation}\label{Eq:SPNEControlStrategies}
    \begin{bmatrix}
    {u^d_n}^*\\
    {u^a_n}^*
    \end{bmatrix} = - M^{-1}_{n}\begin{bmatrix}
    {B^d}' \\ {B^a}'
    \end{bmatrix} L_{n+1} A \hat{x}_n,
\end{equation}
where the estimator $\hat{x}_n = \mathbb{E}\left[x_n\middle \vert \bar{\mathcal{F}}_n \right]$ is given by a Kalman-type linear filter \cite{kalman1960new} operating on the observation data $Y_n$ decided by $I_n^d$ and $I_n^a$. The covariance of the estimation error associated with the filter $P_n(\bar{\mathcal{F}}_n) = \mathbb{E}\left[ (x_n-\hat{x}_n)(x_n - \hat{x}_n)'\right]$ can be propagated as
{\begin{equation}\label{ConvariancePropagation}
\begin{aligned}
    &P_n(\bar{\mathcal{F}}_n) =\\
    &\begin{cases}
    A P_{n-1}(\bar{\mathcal{F}}_{n-1})A' + C\Sigma_s C', \quad \quad \quad  \textrm{if }h(i_n^d, i^a_n) = 0,\\
    \left(\Id - G_n D\right) \left( A P_{n-1}(\bar{\mathcal{F}}_{n-1})A' + C \Sigma_s C'\right)\times \\
    \left( \Id - G_n D \right)'+ G_n E \Sigma_oE'G_n',\quad \ \ \  \textrm{if }h(i_n^d, i^a_n) = 1,\\
    \end{cases}
\end{aligned}
\end{equation}}where $G_n$ can be recognized as one of the usual Kalman filter gains with
$$
\begin{aligned}
G_n =& \left( AP_{n-1}(\bar{\mathcal{F}}_{n-1}) A' + C \Sigma_s C' \right)D'\times\\
&\left[ D\left(AP_{n-1}(\bar{\mathcal{F}}_{n-1})A'+ C\Sigma_sC' \right)D' + E\Sigma_o E'\right]^{-1}. 
\end{aligned}
$$
Here, the observation effect coefficients $\varphi_n$ in \Cref{Eq:SPNEControlCost2Go} is given as
\begin{equation}\label{Eq:ObservationEffectCoefficient}
\varphi_n = A' L_{n+1} \begin{bmatrix}
B^d & B^a
\end{bmatrix} M^{-1}_{n} 
\begin{bmatrix}
{B^d}'\\
{B^a}'
\end{bmatrix} L_{n+1}A.
\end{equation}

\end{theorem}
\proof
The proof is conducted by backward induction. When $k=N$, there are no control strategies involved at this stage. Hence, we have
$$
{V_N^d}^*(\bar{\mathcal{F}}_N) = V_N^d(\bar{\mathcal{F}}_N) = \mathbb{E}\left[x_N' Q_Nx_N\middle\vert \bar{\mathcal{F}}_N \right].
$$
which agrees with \Cref{Eq:SPNEControlCost2Go}. We demonstrate that \Cref{,Eq:SPNEControlCost2Go,Eq:RicattiUpdate,Eq:SPNEControlStrategies,Eq:SPNEControlAuxiliaryMatrix,ConvariancePropagation,Eq:ObservationEffectCoefficient} hold when $k=N-1$. By definition, we have 
\begin{equation}\label{Eq:ProofCost2Go-1}
\begin{aligned}
{V_{N-1}^d}^*&=\min_{u_{N-1}^d} \max_{u_{N-1}^a} \mathbb{E}\Big[ x_{N-1}' Q_{N-1}x_{N-1} + {u^{d}_{N-1}}'R_{N-1}^d u^{d}_{N-1}\\
&- {u^{a}_{N-1}}'R_{N-1}^a u^{a}_{N-1} + i_{N-1}^d O_{N-1}^d - i_{N-1}^a O_{N-1}^a\\
&+ x_N' L_N x_N   \Big\vert \bar{\mathcal{F}}_{N-1} \Big].
\end{aligned}
\end{equation}
Substituting $X_N = Ax_{N-1} + B^d u_{N-1}^d + B^au^a_{N-1} +Cw_{N-1}$, carrying out the expectation, minimizing over $u_{N-1}^d$, and maximizing over $u^a_{N-1}$ yield the SPNE control for this stage
{\footnotesize
\begin{equation}\label{Eq:SPNEControlLinearEquations}
\begin{aligned}
{u^d_{N-1}}^* = - (R_{N-1}^d + {B^d}'L_N B^d)^{-1}{B^d}' L_N(A\hat{x}_{N-1}+ B^a {u^a_{N-1}}^*),\\
{u^a_{N-1}}^* = - ({B^a}'L_N B^a -R_{N-1}^a)^{-1}{B^a}' L_N(A\hat{x}_{N-1}+ B^d {u^d_{N-1}}^*).
\end{aligned}
\end{equation}}
Solving the two linear equations yields \Cref{Eq:SPNEControlStrategies} for the case of $n=N-1$. Substituting the SPNE control back into \Cref{Eq:ProofCost2Go-1} gives
\begin{equation}\label{Eq:ProofCost2Go-2}
\begin{aligned}
&{V_{N-1}^d}^*=\\
& \mathbb{E}\Big[ x_{N-1}' \left(Q_{N-1} + A'L_NA \right)x_{N-1}\\
&- \hat{x}_{N-1}' A' L_N \begin{bmatrix}
B^d & B^a
\end{bmatrix}M_{N-1}^{-1} \begin{bmatrix}
{B^d}' \\ {B^a}'
\end{bmatrix}L_N A \hat{x}_{N-1}\\
& + i_{N-1}^d O_{N-1}^d - i_{N-1}^a O_{N-1}^a \Big\vert \bar{\mathcal{F}}_{N-1} \Big] + \Tr\Sigma_s C' L_N C,
\end{aligned}
\end{equation}
where we have used the fact that $\mathbb{E}\left[w_{N-1}' C' L_N C w_{N-1}\right] = \Tr \Sigma_s C' L_{N-1}C$. The expectation in \Cref{Eq:ProofCost2Go-2} must still be specified for the quadratic term involving $\hat{x}_{n-1}$. For any random vector $x$ and any appropriately dimensioned matrix $M$, we have the relation
$$
\mathbb{E}\left[x' M x\right] = \bar{x}' M \bar{x} + \mathbb{E}\left[(x-\bar{x})' M(x-\bar{x})\right],
$$
where $\bar{x} = \mathbb{E}\left[ x\right]$. Applying this relation to the quadratic term involving $\hat{x}_{n-1}$ \Cref{Eq:ProofCost2Go-2} yields
$$
\begin{aligned}
&{V_{N-1}^d}^*=\\
& \mathbb{E}\Big[ x_{N-1}' L_{N-1} x_{N-1} + (x_{N-1}-\hat{x}_{N-1})' A' L_N \begin{bmatrix}
B^d & B^a
\end{bmatrix}\times\\
&M_{N-1}^{-1} \begin{bmatrix}
{B^d}' \\ {B^a}'
\end{bmatrix}L_N A (x_{N-1}-\hat{x}_{N-1}) \Big\vert \bar{\mathcal{F}}_{N-1} \Big] + i_{N-1}^d O_{N-1}^d\\
&  - i_{N-1}^a O_{N-1}^a  + \Tr\Sigma_s C' L_N C.
\end{aligned}
$$
Using the definition of $P_{N-1}(\bar{\mathcal{F}}_{N-1})$ gives
$$
\begin{aligned}
&{V_{N-1}^d}^*= \mathbb{E}\Big[ x_{N-1}' L_{N-1} x_{N-1}\Big\vert \bar{\mathcal{F}}_{N-1} \Big]  + i_{N-1}^d O_{N-1}^d\\
&  - i_{N-1}^a O_{N-1}^a  + \Tr\Sigma_s C' L_N C + \Tr P_{N-1}(\bar{\mathcal{F}}_{N-1})\varphi_{N-1},
\end{aligned}
$$
which agrees with \Cref{Eq:SPNEControlCost2Go}, and $\varphi_{N-1}$ satisfies \Cref{Eq:ObservationEffectCoefficient}. The propagation of $P_{n}(\bar{\mathcal{F}}_n)$ in \Cref{ConvariancePropagation} follows the results of \cite{cooper1971optimal}. Thus, we have shown that \Cref{,Eq:SPNEControlCost2Go,Eq:RicattiUpdate,Eq:SPNEControlStrategies,Eq:SPNEControlAuxiliaryMatrix,ConvariancePropagation,Eq:ObservationEffectCoefficient} hold for $k=N-1$. Suppose that the claims \Cref{,Eq:SPNEControlCost2Go,Eq:RicattiUpdate,Eq:SPNEControlStrategies,Eq:SPNEControlAuxiliaryMatrix,ConvariancePropagation,Eq:ObservationEffectCoefficient} hold for an arbitrary $k+1\leq N$. 
By definition of $V_k^d$ and the tower property of conditional expectation \cite{jacod2012probability}, we have 
$$
\begin{aligned}
V_k^d(\nu^d,\nu^a)  = \mathbb{E}&\left[( x_k' Q_k x_k + {u_k^d}'R_k^d u_k^d - {u^a_k}' R_k^a u_k^a   \right.\\
&+ i^d_k O^d_k + i^a_k O^a_k ) + V_{k+1}^d(\bar{\mathcal{F}}_{k+1}) \bigg\vert \bar{\mathcal{F}}_k  \bigg].
\end{aligned}
$$
An application of dynamic programming techniques yields
$$
\begin{aligned}
&{V_k^d}^*(\bar{\mathcal{F}}_k)\\
=& \min_{u_k^d} \max_{u_k^a} \mathbb{E}\left[ x_k' Q_k x_k + {u_k^d}'R_k^d u_k^d - {u^a_k}' R_k^a u_k^a  + i^d_k O^d_k\right.\\
&+ i^a_k O^a_k  + {V_{k+1}^d}^*(\bar{\mathcal{F}}_{k+1}) \bigg\vert \bar{\mathcal{F}}_k  \bigg]\\
=&\sum_{n=k+1}^{N-1} \left[\Tr \left(P_n(\bar{\mathcal{F}}_n)\varphi_n+\Sigma_s C' L_{n+1} C\right)+i_n^d O_n^d- i_n^a O_n^a\right]\\
&+\min_{u_k^d}\max_{u_k^a}\mathbb{E}\Bigg[ x_k' Q_k x_k + {u_k^d}'R_k^d u_k^d - {u^a_k}' R_k^a u_k^a + i^d_k O^d_k\\
&+ i^a_k O^a_k + x_{k+1}' L_{k+1} x_{k+1}
\Big\vert \bar{\mathcal{F}}_k\Big].
\end{aligned}
$$
The remaining proof, which deals the minimax term, is identical to the proof for the case when $k =N-1$. Now it remains to show that $M_n$ is invertible when $R_n^a > {B^a}' L_{n+1}{B^a}$. The matrix $M_n$ is invertible, provided that the Schur complement \cite{zhang2006schur}, 
$$
\begin{aligned}
S_C \equiv& {B^a}' L_{n+1} B^a - R^a_{n} - {B^a}' L_{n+1} B^d\times\\
&(R_{n}^d + {B^d}'L_{n+1}B^d)^{-1}{B^d}'L_{n+1}B^a,
\end{aligned}
$$
is invertible. Clearly, the Schur complement $S_C$ is a real, symmetric matrix. And $R_n^a > {B^a}' L_{n+1}{B^a}$ guarantees that the Schur complement $S_C$ is negative definite, hence invertible.
\endproof

\begin{remark}
It is required to calculate the matrix inverse $M_n^{-1}$. The Schur complement of the bottom-right-corner block in the $M_n$ matrix is the real, symmetric matrix
$$
\begin{aligned}
S_B(L_{n+1}) &= R_{n}^d + {B^d}'L_{n+1}B^d + {B^d}'L_{n+1}B^a \times \\
&(R^a_{n} - {B^a}' L_{n+1} B^a)^{-1} {B^a}'L_{n+1}B^d,
\end{aligned}
$$
which is positive definite since $R^a_{n} > {B^a}' L_{n+1} B^a$, and hence invertible. Therefore, the matrix $M_n^{-1}$ can be factored as follows:
\begin{equation}\label{Eq:MDecomposition}
M_n^{-1}  = \Omega T\Omega'
\end{equation}
with
\begin{equation}\label{Eq:MDecompositionComponents}
\begin{aligned}
\Omega' &= \begin{bmatrix}
\Id & {B^d}' L_{n+1} B^a \left[R^a_n - {B^a}'L_{n+1}B^a\right]^{-1}\\
0 & \Id
\end{bmatrix},\\
T & =\begin{bmatrix}
S_B^{-1}(L_{n+1}) & 0\\
0 & -\left[R^a_n - {B^a}' L_{n+1} B^a\right]^{-1}
\end{bmatrix}.
\end{aligned}
\end{equation}
This allows the defender and the attacker to compute their SPNE control laws using explicit formulae,
$$
{\small
\begin{aligned}
{u^d_n}^* =& \nu_n^d(\bar{\mathcal{F}}_n) = - S_B^{-1}(L_{n+1}){B^d}'\times\\ &\Big\{\Id + L_{n+1}B^a \big[R_n^a - {B^a}'L_{n+1} B^a\big]^{-1}{B^a}'\Big\} L_{n+1} A \hat{x}_n\\
{u^a_n}^* =& \nu_a^d(\bar{\mathcal{F}}_n) = [R_n^a - {B^a}'L_{n+1} B^a]^{-1} {B^a}'\times\\
&\bigg( \Id - L_{n+1} B^d S^{-1}_B(L_{n+1}){B^d}' \Big\{\Id + L_{n+1} B^a \times\\
&\left[R_n^a - {B^a}'L_{n+1} B^a\right]^{-1} {B^a}'\Big\} \bigg)L_{n+1}A\hat{x}_n.
\end{aligned}
}
$$
\end{remark}

\begin{remark}
The assumptions that $R^a_n > {B^a}'L_{n+1}B^a$ is not stringent in the setting of adversarial attacks since the cost of injecting malicious controls into the plant is usually expensive, much higher than the normal controls implemented by the defender.
\end{remark}

\begin{remark}
If it is assumed that the initial state $x_0$ is known to both players, the following initial conditions must hold for $P_0(\bar{\mathcal{F}}_0)$ in \Cref{ConvariancePropagation}: $i_0^d = i_0^a =0$ and $P_0(\bar{\mathcal{F}}_0) =0$. This is because the exact knowledge of the initial state $x_0$ at no cost is preferable for the defender to a noisy and expensive initial observation. Hence, this is no incentive to observe for the defender, and as a result, there is no need to jam. If $x_0$ is Gaussian distributed with mean $\bar{x}_0$ and variance $\Sigma_0$, whose realization is unknown to both players but the statistics are known to both, then the following initial conditions hold in \Cref{ConvariancePropagation}:
{
$$
\begin{aligned}
&P_0(\bar{\mathcal{F}}_0) =\\
&\begin{cases}
\Sigma_0,\ \ &\textrm{if }h(i^d_n,i^a_n) =0,\\
\Sigma_0 - \Sigma_0 D'[D'\Sigma_0 D + E'\Sigma_o E]^{-1}D'\Sigma_0,\ \ &\textrm{if }h(i^d_n,i^a_n) =1.
\end{cases}\\
\end{aligned}
$$}
\end{remark}

Given $I_{N-1}^d$ and $I_{N-1}^a$, the total expected cost conditioned on $\bar{\mathcal{F}_0}$ for the defender with the SPNE control strategies is ${V_0^d}^*$. The cost includes the quadratic term $\mathbb{E}\left[x_0'L_0x_0\middle\vert \bar{\mathcal{F}}_0\right]$, the accumulated cost induced by system noise $\sum_{n=0}^{N-1} \Tr \Sigma_s C'L_{n+1}C$, the accumulated cost induced by the estimation error $\sum_{n=0}^{N-1} \Tr P_n(\bar{\mathcal{F}}_n) \varphi_n$, where the estimation error relies heavily on the observation and the jamming decisions, as well as the accumulated costs of observing $\sum_{n=0}^{N-1} i_n^d O_n^d$ and that of jamming $\sum_{n=0}^{N-1} i_n^a O_n^a$.

Being aware of each other's strategies as shown by \Cref{Eq:SPNEControlLinearEquations}, the attacker and the defender deploy their SPNE control strategies according to linear control laws based on their estimate of the state. The jamming of the defender's observation can undermine the control performance, but at the same time, it also impairs the attacking performance at the physical layer since the attacker also suffers from less information, let alone the jamming cost $O^a_n$. Whether the defender should observe or not also depends on multiple factors including the cost of observation $O^d_n$, the control performance degradation caused by the estimation error $\Tr P_n(\bar{\mathcal{F}}_n) \varphi_n$, and the implicit cost of offering more information to the attacker. We will shed more ligth on the observation and the jamming strategies in \Cref{Subsec:ObservationStrategies}.

\subsection{The Observation and the Jamming Strategies}\label{Subsec:ObservationStrategies}

In \Cref{subsec:SPNEControlStrategies}, we have demonstrated that for any given observation sequences, the expected cost-to-go of the game under the SPNE control strategies after the observation and the jamming decision have been taken at time $k$ is ${V_k^d}^*$. In this section, we derive the the procedures for finding the SPNE observation and jamming strategies by leveraging the results we have developed in \Cref{subsec:SPNEControlStrategies}.

Note that instead of the inner cost-to-go $V_k^d(\bar{\mathcal{F}}_k)$, the cost-to-go before the observation and the jamming decisions are made at stage $k$ is $f_k^d(\mathcal{F}_k)$ defined in \Cref{Eq:ExpectedCosttoGo}.  The strategies to be made are the observation strategy $\mu^d_n$, the jamming strategy $\mu^a_n$, and the control strategies $\nu_n^d$ and $\nu_n^a$ for all $n\geq k$. By the definition of $V_k^d(\bar{\mathcal{F}}_k)$ in \Cref{Eq:DefinitionPostObservationCost2Go}, the definition of $f_k^d$ in \Cref{Eq:ExpectedCosttoGo} and the fact that $\mathcal{F}_k \subset \bar{\mathcal{F}}_k$, we have
$
f_k^d(\mathcal{F}_k) = \mathbb{E}\left[ V_k^d(\bar{\mathcal{F}}_k)\middle\vert \mathcal{F}_k\right].
$
The defender aims to find both the observation strategy and the control strategy that minimize $f_k^d(\mathcal{F}_k)$ at every stage $k$ given $\mathcal{F}_k$ while the attacker aims to do the opposite. Since the control at time stage $n\geq k$ dose not alter the information $\mathcal{F}_k$, with a slight abuse of notation, we can write
$$
{f_k^d}^* (\mathcal{F}_k) = \min_{\pi^d} \max_{\pi^a} f_k^d(\mathcal{F}_k)= \min_{\mu^d}\max_{\mu^a} \mathbb{E}\left[{V_k^d}^* \middle\vert \mathcal{F}_k\right] 
$$

Using \Cref{ConvariancePropagation}, \Cref{Eq:SPNEControlCost2Go} can be written 
$$
\begin{aligned}
{V_k^d}^* =
&\mathbb{E}\left[x_k' L_{k} x_k\middle\vert \bar{\mathcal{F}}_k\right] + \sum_{n=k}^{N-1} \Tr \Sigma_s C'L_{n+1}C\\
&+\sum_{n=k}^{N-1} \Tr \big[ Z\left(P_{n-1}(\bar{\mathcal{F}}_{n-1})\right) - h(i^d_n, i_n^a)\times\\
&\quad H(P_{n-1}\left(\bar{\mathcal{F}}_{n-1})\right)\big]\varphi_n + i_n^d O_n^d -i_n^a O_n^a,
\end{aligned}
$$
for $k\geq 1$, where
$$
\begin{aligned}
Z(P) &= APA' + C\Sigma_s C',\\
H(P) &= Z(P)D'[D(Z(P))D'+ E\Sigma_o E']^{-1}DZ(P).
\end{aligned}
$$
We define $F_k^d(\mathcal{F}_k) = \mathbb{E} \left[ {V_k^d}^*|\mathcal{F}_k \right]$. Using the tower property yields
$$
F_k^d(\mathcal{F}_k) = K_k^d(\mathcal{F}_k) + J_k^d(\mu^d,\mu^a,\mathcal{F}_k),
$$
where 
$$
\begin{aligned}
K_k^d(\mathcal{F}_k) &= \mathbb{E}\left[ x_k'L_k x_k \middle\vert \mathcal{F}_k \right] +  \sum_{n=k}^{N-1} \Tr\Sigma_s C' L_{n+1}C,\\
J_k^d(\mu^d,\mu^a,\mathcal{F}_k) &= \sum_{n=k}^{N-1} \Tr \big[ Z\left(P_{n-1}(\bar{\mathcal{F}}_{n-1})\right) - h(i^d_n, i_n^a)\times\\
&\quad H(P_{n-1}\left(\bar{\mathcal{F}}_{n-1})\right)\big]\varphi_n + i_n^d O_n^d -i_n^a O_n^a.
\end{aligned}
$$
Therefore, to find the SPNE observation strategies, one can only focus on $J_k^d(\mu^d,\mu^a,\mathcal{F}_k)$, which we write as $J_k^d(\mathcal{F}_k)$ for convenience in later discussion. Let us define the quantity ${J_k^d}^*(\mathcal{F}_k) = \min_{\mu^d} \max_{\mu^a} J_k^d(\mathcal{F}_k)$.

If the SPNE observation and the jamming strategies have been given for every possible $\mathcal{F}_{k+1}$, then the rule for selecting the Nash equilibrium at stage $k$,
{$$
\begin{aligned}
{J_k^d}^*(\mathcal{F}_k) =& \min_{i_k^d} \max_{i_k^a} \Tr \big[ Z\left(P_{k-1}(\bar{\mathcal{F}}_{k-1})\right) - h(i^d_k, i_k^a)\times\\ &H(P_{k-1}\left(\bar{\mathcal{F}}_{k-1})\right)\big]\varphi_k + i_k^d O_k^d -i_k^a O_k^a\\
&+ {J^d_{k+1}}^*(\mathcal{F}_{k+1}).\,
\end{aligned}
$$}
for $k=1,2,\cdots,N-1$ and we have ${J_N^d}^* = 0$ by definition.

\begin{Proposition}\label{Prop:EquilibriumMayNotExist}
Suppose that there is a sequence of SPNE observation and jamming strategies $\{({\mu^d_n}^*,{\mu_n^a}^*) ,n=k+1,\cdots,N-1 \}$. Then ${J_{k+1}^d}^*$ can be expressed as a function $P_k(\bar{\mathcal{F}}_k)$. There exists no Nash equilibrium at stage $k$ if 
\begin{equation}\label{Eq:ConditionsNoEquilibrium}
\begin{aligned}
O_{k}^a \leq & \Tr H(P_{k-1})\varphi_k + {J_{k+1}^d}^*(Z(P_{k-1}))\\ &-{J_{k+1}^d}^*(Z(P_{k-1})- H(P_{k-1})),\\
O_k^d \leq & \Tr H(P_{k-1})\varphi_k + {J_{k+1}^d}^*(Z(P_{k-1}))\\ &-{J_{k+1}^d}^*(Z(P_{k-1})- H(P_{k-1})),\\
\end{aligned}
\end{equation}
where $P_{k-1} \coloneqq P_{k-1}(\bar{\mathcal{F}}_{k-1})$ for simplicity.
\end{Proposition}
\begin{proof}
At stage $k=N-1$, we have 
{\small$$
\begin{aligned}
{J^d_{N-1}}^* (\mathcal{F}_{N-1}) =&\min_{i_{N-1}^d}\max_{i_{N-1}^a} \Tr \big[ Z\left(P_{N-2}\right) - h(i^d_{N-1}, i_{N-1}^a)\times \\ &H(P_{N-2})\big] \varphi_{N-1} + i_{N-1}^d O_{N-1}^d -i_{N-1}^a O_{N-1}^a.
\end{aligned}
$$}
Since the observation decision $i^d_{N-1}$ and the jamming decision $i^a_{N-1}$ are both binary, the game at stage $N-1$ can be written in the form of a matrix game, which is given in \Cref{tab:PayoffMatrixStageN-1}. An simple analysis of the matrix game shows that if
$
O_{N-1}^d \geq \Tr H(P_{N-2})\varphi_{N-1},
$
the Nash equilibrium of the matrix game is ${i_{N-1}^d}^*= 0$, ${i_{N-1}^a}^*= 0$; if 
$$
\begin{aligned}
O_{N-1}^d &< \Tr H(P_{N-2})\varphi_{N-1} \leq 
O_{N-1}^a ,
\end{aligned}
$$
there exists a Nash equilibrium of the matrix game; i.e., ${i_{N-1}^d}^*= 1$, ${i_{N-1}^a}^*= 0$;
there does not exist a Nash equilibrium if 
$$
\begin{aligned}
O_{N-1}^d &< \Tr H(P_{N-2})\varphi_{N-1},\ \ \ \textrm{and}\\
O_{N-1}^a &< \Tr H(P_{N-2})\varphi_{N-1},
\end{aligned}
$$
which aligns with \Cref{Eq:ConditionsNoEquilibrium}. Note that when the Nash equilibrium does exist, $J_{N-1}^*(\mathcal{F}_{N-1})$ depends only on $P_{N-2}$ except other system parameters and cost coefficients that are common knowledge to both players. Suppose that the claims in Proposition \Cref{Prop:EquilibriumMayNotExist} hold for $k+1$. For stage $k$, we have
\begin{equation}\label{Eq:DynamicProgrammingObservation}
\begin{aligned}
{J^d_{k}}^* (\mathcal{F}_{k}) =&\min_{i_{k}^d}\max_{i_{k}^a} \Tr \big[ Z\left(P_{k-1}\right) - h(i^d_{k}, i_{k}^a) H(P_{k-1})\big] \varphi_{k}\\
& + i_{N-1}^d O_{N-1}^d -i_{N-1}^a O_{N-1}^a + {J_{k+1}^d}^*(P_k)\\
=& \min_{i_{k}^d}\max_{i_{k}^a} \Tr \big[ Z\left(P_{k-1}\right) - h(i^d_{k}, i_{k}^a) H(P_{k-1})\big] \varphi_{k}\\
& + i_{N-1}^d O_{N-1}^d -i_{N-1}^a O_{N-1}^a\\
 &+ {J_{k+1}^d}^*\left(Z\left(P_{k-1}\right) - h(i^d_{k}, i_{k}^a) H(P_{k-1})\right).\\
\end{aligned}
\end{equation}
The remaining proof follows the same procedures in the case of $k=N-1$.
{	
\begin{table}[H]
\caption{The payoff matrix for the zero-sum game between the defender and the attacker at stage $N-1$.}
\label{tab:PayoffMatrixStageN-1}
\begin{tabular}{|c|c|c|c|}
\hline
\multicolumn{2}{|c|}{\multirow{2}{*}{$J_{N-1}^d$}} & \multicolumn{2}{c|}{$i^a_{N-1}$}                                                                                          \\ \cline{3-4} 
\multicolumn{2}{|c|}{}                           & 1                                                                                     & 0                             \\ \hline
\multirow{2}{*}{$i^d_{N-1}$} &
  1 &
  \begin{tabular}[c]{@{}c@{}}$\Tr Z(P_{N-2})\varphi_{N-1}$\\ $+O_{N-1}^d - O_{N-1}^a$\end{tabular} &
  \begin{tabular}[c]{@{}c@{}}$\Tr [Z(P_{N-2}) - H(P_{N-2})]\times$\\$\varphi_{N-1}$\\ $+O_{N-1}^d$\end{tabular} \\ \cline{2-4} 
                       & 0                       & \begin{tabular}[c]{@{}c@{}}$\Tr Z(P_{N-2})\varphi_{N-1}$\\ $- O_{N-1}^a$\end{tabular} & $\Tr Z(P_{N-2})\varphi_{N-1}$ \\ \hline
\end{tabular}
\end{table} 
}
\end{proof}

We assume in the proof that when the margin is zero, there is no incentive for both players to act. The defender and the attacker have to decide at each stage whether to observe and to attack respectively yet simultaneously. When the conditions in \Cref{Eq:ConditionsNoEquilibrium} hold, there is an incentive to observe if the attacker does not jam but there always an incentive for the attacker to jam if the defender observes. Hence, there exists no Nash equilibrium in pure strategy. Now, suppose that the defender announces its observation decision first, then the attacker chooses whether to jam. That is to say at stage $n$ the observation and the jamming strategies can be written as $\mu^d_n(\mathcal{F}_{n}^d)$ and $\mu^a_n(\mathcal{F}_{n}^a)$, where
$$
\begin{aligned}
\mathcal{F}_{n}^d = \mathcal{F}_n,\ \ \mathcal{F}_{n}^a = \mathcal{F}_n \cup I_{n}^d.
\end{aligned}
$$

\begin{theorem}\label{Theorem:ObservationDynamicProgramming}
Under the information structure $\mathcal{F}_n^d$ and $\mathcal{F}_n^a$, for any stage $k\geq 1$, there always exist a pair of Subgame Perfect Equilibrium (SPE) strategies that depend only on $P_{k-1}$. The equilibrium at stage $k$ is $({i_k^d}^* , {i_k^a}^*)=(1,1)$ if
\begin{equation}\label{Eq:ThresholdObservationPolicy}
\begin{aligned}
O_k^d <& O_k^a < \Tr H(P_{k-1})\varphi_k + {J_{k+1}^d}^*(Z(P_{k-1})) \\
&-{J_{k+1}^d}^*(Z(P_{k-1})-H(P_{k-1}));\\
\end{aligned}
\end{equation}
The equilibrium is $({i_k^d}^* , {i_k^a}^*)=(1,0)$ if
\begin{equation}\label{Eq:ThresholdObservationPolicy1}
\begin{aligned}
O_k^d <& \Tr H(P_{k-1})\varphi_k + {J_{k+1}^d}^*(Z(P_{k-1})) \\
&-{J_{k+1}^d}^*(Z(P_{k-1})-H(P_{k-1}))\leq O_k^a;\\
\end{aligned}
\end{equation}
and $({i_k^d}^* , {i_k^a}^*)=(0,0)$ otherwise. Hence, ${J_{k}^d}^*$ also depends solely on $P_{k-1}$.
\end{theorem}
\begin{proof}
It is easy to see the hypothesis holds at stage $N-1$. Suppose that the hypothesis is true for stage $k+1$. Following the same argument in the proof of \Cref{Prop:EquilibriumMayNotExist}, we can arrive at the same matrix described by \Cref{tab:MatrixGameStagek} except that now the defender announces its observation decision first, then the defender reacts. In this circumstance, we have a Stackelberg game and a Stackelberg equilibrium. When the defender chooses not to observe, the best response of the attacker is not to jam since jamming generates no benefit but additional cost of jamming. This scenario gives $(i^d_k,i_k^a) = (0,0)$ and cost-to-go $J_k^d = \Tr Z(P_{k-1})\varphi_{k}{+J_{k+1}^d}^*(Z(P_{k-1}))$. When the defender chooses to observe, the best response of the attacker is to jam if $O_k^a  < \Tr H(P_{k-1})\varphi_k + {J_{k+1}^d}^*(Z(P_{k-1})) -{J_{k+1}^d}^*(Z(P_{k-1})-H(P_{k-1}))$, which gives cost-to-go $J_k^d = \Tr Z(P_{k-1})\varphi_{k} +O_{k}^d - O_{k}^a {+J_{k+1}^d}^*(Z(P_{k-1}))$. The best response becomes not jamming if $O_k^a  \geq \Tr H(P_{k-1})\varphi_k + {J_{k+1}^d}^*(Z(P_{k-1})) -{J_{k+1}^d}^*(Z(P_{k-1})-H(P_{k-1}))$, which produces cost-to-go $\Tr [Z(P_{k-1}) - H(P_{k-1})]\times\varphi_{k}+O_{k}^d {+J_{k+1}^d}^*(Z(P_{k-1}) - H(P_{k-1}))$. The defender then makes appropriate observation decision that generates the least cost-to-go by anticipating the best response of the attacker. Following this logic, we obtain that the equilibrium at stage $k$ is $({i_k^d}^* , {i_k^a}^*)=(1,1)$ if
$$
\begin{aligned}
O_k^d <& O_k^a < \Tr H(P_{k-1})\varphi_k + {J_{k+1}^d}^*(Z(P_{k-1})) \\
&-{J_{k+1}^d}^*(Z(P_{k-1})-H(P_{k-1}));\\
\end{aligned}
$$
The equilibrium is $({i_k^d}^* , {i_k^a}^*)=(1,0)$ if
$$
\begin{aligned}
O_k^d <& \Tr H(P_{k-1})\varphi_k + {J_{k+1}^d}^*(Z(P_{k-1})) \\
&-{J_{k+1}^d}^*(Z(P_{k-1})-H(P_{k-1}))\leq O_k^a;\\
\end{aligned}
$$
and $({i_k^d}^* , {i_k^a}^*)=(0,0)$ otherwise. 
\begin{table}[H]
\caption{The payoff matrix for the zero-sum game between the defender and the attacker at stage $k\geq1$.}
\label{tab:MatrixGameStagek}
\begin{tabular}{|c|c|c|c|}
\hline
\multicolumn{2}{|c|}{\multirow{2}{*}{$J_{k}^d$}} & \multicolumn{2}{c|}{$i^a_k$} \\ \cline{3-4} 
\multicolumn{2}{|c|}{}                         & 1             & 0            \\ \hline
\multirow{2}{*}{$i^d_k$} &
  1 &
  \begin{tabular}[c]{@{}c@{}}$\Tr Z(P_{k-1})\varphi_{k}$\\ $+O_{k}^d - O_{k}^a$\\ ${+J_{k+1}^d}^*(Z(P_{k-1}))$\end{tabular} &
  \begin{tabular}[c]{@{}c@{}}$\Tr [Z(P_{k-1}) - H(P_{k-1})]\times$\\ $\varphi_{k}+O_{k}^d$\\ ${+J_{k+1}^d}^*(Z(P_{k-1}) - H(P_{k-1}))$\end{tabular} \\ \cline{2-4} 
 &
  0 &
  \begin{tabular}[c]{@{}c@{}}$\Tr Z(P_{k-1})\varphi_{k}$\\ $- O_{k}^a$\\ $+{J_{k+1}^d}^*(Z(P_{k-1}))$\end{tabular} &
  \begin{tabular}[c]{@{}c@{}}$\Tr Z(P_{k-1})\varphi_{k}$\\ ${+J_{k+1}^d}^*(Z(P_{k-1}))$\end{tabular} \\ \hline
\end{tabular}
\end{table}
Hence, the strategies at stage $k$ depend solely on $P_{k-1}$. And ${J_k^d}^*$ is a function of $P_{k-1}$. Here, we assume that there is no incentive to act when the margin between two actions is zero.
\end{proof}
Even though in some circumstances, the defender will be better off if she/he can receive the observation, she/he will not observe to avoid additional cost of observation since she/he can anticipate the observation being jammed.

\begin{corollary}\label{cor:Attacker}
Suppose that at each stage $k$, the attacker announces her/his jamming decision first, then the defender reacts; i.e., $\mathcal{F}^a_k = \mathcal{F}_k$ and $\mathcal{F}_k^d = \mathcal{F}_k \cup  I_{k}^a$. The equilibrium at stage $k$ is  $({i_k^d}^* , {i_k^a}^*)=(0,0)$ if
$$
\begin{aligned}
O_k^d \geq& \Tr H(P_{k-1})\varphi_k + {J_{k+1}^d}^*(Z(P_{k-1})) \\
&-{J_{k+1}^d}^*(Z(P_{k-1})-H(P_{k-1}));
\end{aligned}
$$
the equilibrium is $({i_k^d}^* , {i_k^a}^*)=(0,1)$
$$
\begin{aligned}
O_k^d + O_k^d <& \Tr H(P_{k-1})\varphi_k + {J_{k+1}^d}^*(Z(P_{k-1})) \\
&-{J_{k+1}^d}^*(Z(P_{k-1})-H(P_{k-1}));\\
\end{aligned}
$$
and $({i_k^d}^* , {i_k^a}^*)=(1,0)$ otherwise.
\end{corollary}
Here, we use the same notation ${J_{k}^d}^*$ for games with different information structures. 

\begin{remark}
The dynamic programming equation in \Cref{Eq:DynamicProgrammingObservation} can be thought of as a dynamic programming equation for a deterministic system with state $P_k$ and controls $i^d_k,i^a_k$. The nonlinear state dynamics are $P_{k+1} = Z(P_{k}) - h(i^d_k,i^a_k)\cdot H(P_{k})$ with the stage cost being $\Tr P_{k}\varphi_k^d + i_k^d O_k^d - i_k^a O_k^a$.
\end{remark}

Note that $P_k = Z(P_{k-1}) - H(P_{k-1})$ if the observation is made and not jammed. Since $H(P_{k-1})$ is always positive semi-definite, at stage $k$, if the observation is missing, i.e., $h(i^d_k,i^a_k) =0$ and $P_k = Z(P_{k-1})$, the covariance of estimate error $P_{k}$ will be larger than the covariance when there is an observation received. Here, by saying the covariance of estimate error $P_k$ is larger than the covariance of estimate error $P_k'$, we mean $P_k \geq P_k'$. Note that the instantaneous cost at stage $k$,\ is $\Tr P_{k}\varphi_k \equiv \mathbb{E}\left[ (x_k - \hat{x}_k)'\varphi_k (x_k - \hat{x}_k) \right]$. From \Cref{Eq:ObservationEffectCoefficient,Eq:MDecomposition}, we know
$$
\varphi_k = A' L_{k+1} \begin{bmatrix}
B^d & B^a
\end{bmatrix} \Omega T \Omega' 
\begin{bmatrix}
{B^d}'\\
{B^a}'
\end{bmatrix} L_{k+1}A,
$$
where $T$, as we can see in \Cref{Eq:MDecompositionComponents}, is a block diagonal matrix with two blocks $S_B^{-1}(L_{k+1})$ and $-(R_n^a - {B^a}' L_{n+1} B^a)^{-1}$. Since the former block is positive definite and the later is negative definite,  $\varphi_k$ is neither positive semi-definite nor negative semi-definite. That means $\Tr P_{k}\varphi_k$ could be negative. The interpretation is that a larger estimate error may not be always detrimental to the defender in the presence of attacks since the observation can help the attacker inject a better control into the physical plant. Thus, even when the cost of observation $O_k^d, k=0,1,\cdots,N-1$ is zero, the defender will not have incentives to observe at each stage. However, when $R^a_k \rightarrow \infty$ which means that there will be no attacks on the physical plant due to high costs, $\varphi_k$ becomes positive semi-definite and in this circumstance, the defender will favor a smaller covariance.

A special case of our framework is when the defender can obtain perfect observations. That is when $D = \Id$ and $E =0$ or $\Sigma_o = 0$. In this circumstance, the covariance of the estimate error propagates as
{
$$
\begin{aligned}
P_{n}(\bar{\mathcal{F}}_{n}) = \begin{cases}
A P_{n-1}(\bar{\mathcal{F}}_{n-1}) A' + C\Sigma_s C',&\textrm{if }h(i_n^d,i_n^a) = 0,\\
0,\ \ &\textrm{if }h(i_n^d,i_n^a) = 1,
\end{cases}
\end{aligned}
$$
with the initial condition 
$$
\begin{aligned}
P_{0}(\bar{\mathcal{F}}_{0}) = \begin{cases}
\Sigma_0,&\textrm{if }h(i_0^d,i_0^a) = 0,\\
0,\ \ &\textrm{if }h(i_0^d,i_0^a) = 1.
\end{cases}
\end{aligned}
$$
}
The rest of the derivation will follow the results we have developed for the noisy observation case.

\subsection{Computation}
From \Cref{Subsec:ObservationStrategies}, we know that the SPNE observation strategies can be derived from a sequence of dynamic programming equations. However, it is impossible to obtain a closed-form expression of the observation value function ${J_k^d}^*$ for all $k$ due to the nonlinearity of the covariance propagation (\ref{ConvariancePropagation}) and the discrete strategies of observing and jamming at each stage. In this section, we consider two computational methods the compute the SPNE observation and jamming strategies. One is to use policy iteration type of algorithms and another is to take decisions online by leveraging pre-computed values. We focus on the case when the defender announces her/his observation decision first.

\subsubsection{Policy Iteration Type of Algorithms}\label{subsubsec:PolicyIteration}

This computational method follows the idea of policy iteration \cite{bertsekas1996neuro}. The SPNE observation solution can be obtained by the following procedure.
\begin{enumerate}
    \item Pick arbitrary observation and jamming sequences $\tilde{I}_{N-1}^d$ and $\tilde{I}_{N-1}^a$ that do not contain $(i_k^d,i_k^a) = (0,1)$(not an equilibrium for sure) for every $k=0,1,\cdots,N-1$. 
    \item Policy Evaluation: Given $\tilde{I}_{N-1}^d$ and $\tilde{I}_{N-1}^a$, generate a sequence of value functions $\tilde{J}^d_k$ and ${\tilde{J}_k^a}$ according to
    $$
    \begin{aligned}
    {J_{k}^d}^*(P) &= \Tr Z(P)\varphi_{k} +O_{k}^d - O_{k}^a {+J_{k+1}^d}^*(Z(P)),\\
    &\textrm{When } \tilde{i}_k^d =1,\tilde{i}_k^d =1;\\ 
    {J_{k}^d}^*(P) &= \Tr [Z(P)-H(P)]\varphi_{k} +O_{k}^d\\ &{+J_{k+1}^d}^*(Z(P)-H(P)),\\
    &\textrm{When } \tilde{i}_k^d =1,\tilde{i}_k^d =0;\\ 
    {J_{k}^d}^*(P) &= \Tr Z(P)\varphi_{k} {+J_{k+1}^d}^*(Z(P)),\\
    &\textrm{When } \tilde{i}_k^d =0,\tilde{i}_k^d =0.\\ 
    \end{aligned}
    $$
    \item Policy update: Given the value functions ${\tilde{J}_k^d}$ computed in 2), generate new sequences of observation and jamming decisions $I^d_{N-1}$ and $I^a_{N-1}$ according to \Cref{Eq:ThresholdObservationPolicy,Eq:ThresholdObservationPolicy1} in \Cref{Theorem:ObservationDynamicProgramming}.
    \item Replace $\tilde{I}_{N-1}^d$ and $\tilde{I}_{N-1}^a$ with $I_{N-1}^d$ and $I_{N-1}^a$ respectively in 2), and continue until $\tilde{I}_{N-1}^d = I_{N-1}^d$ and $\tilde{I}^a_{N-1}= I_{N-1}^a$.
\end{enumerate}
No claim is made that the convergence of this procedure is guaranteed. But if convergence does occur, the sequences $I_{N-1}^d$ and $I_{N-1}^a$ are equilibrium and the cost-to-go functions ${J_k^d}^*$ obey \Cref{Eq:DynamicProgrammingObservation}.

\subsubsection{Backward Enumeration Algorithm}
Leveraging the results in \Cref{Theorem:ObservationDynamicProgramming}, we can write ${J_k^d}^*(P)$ in the following manner.
{\small
$$
{J_{k}^d}^*(P)=\begin{cases}
\Tr Z(P)\varphi_{k} +O_{k}^d - O_{k}^a {+J_{k+1}^d}^*(Z(P)),\\
\textrm{if }O_k^d < O_k^a < \mathcal{T}_k(P);\\
\Tr [Z(P)-H(P)]\varphi_{k} +O_{k}^d{+J_{k+1}^d}^*(Z(P)-H(P)),\\
\textrm{if }O_k^d < \mathcal{T}_k(P) \leq O_k^a;\\
\Tr Z(P)\varphi_{k} {+J_{k+1}^d}^*(Z(P)),\\
\textrm{otherwise},
\end{cases}
$$}
with ${J_{N}^d}^* = 0$ and $\mathcal{T}_k(P)\coloneqq \Tr H(P)\varphi_k + {J_{k+1}^d}^*(Z(P)) -{J_{k+1}^d}^*(Z(P)-H(P))$. At any given stage $k$, there is only a finite number of distinct covariance $P_k$ for all possible $I_k^d$ and $I_k^a$. Thus, it is sufficient to find the SPNE observations and jamming strategies by computing and storing a finite number of values to characterize the value functions in the absence of a full characterization over the entire space of symmetric positive semi-definite matrices. For example, at stage $k=0$, we only have two possible $P_0$ which are $\Sigma_0$ and $\Sigma_0 - \Sigma_0 D'[D'\Sigma_0 D + E'\Sigma_o E]D'\Sigma_0$. At stage $k=N-1$. there are $2^{N-1}$ possible values of $P_{N-2}$ for all possible $\mathcal{F}_{N-1}$. Hence, the total number of values needed to store for a given game of horizon $N$ is
$$
\sum_{k=0}^{N-1} = 2^{N+1} -2 \leq 2^{N+1}.
$$

If the defender can obtain perfect observation, at stage $k$, there are $k+1$ possible values $P_{k-1}$. The total number of values needed is 
$$
\sum_{k=0}^{N-1} k+1= \frac{N(N-1)}{2}.
$$

\begin{remark}
In order to compute the observation decision online, we require that the maximum number of values needed to be stored is of order $O(2^N)$. When the defender can receive perfect observation, i.e., $D =I$ and $E =0$, the maximum number of values needed to be stored is of order $O(n^2)$. For $N=30$, the total number of values needed to be stored is approximately $10^{10}$. It is apparent that if players receive noisy observation, for problems of any appreciable size, the amount of storage required can easily become prohibitive. One can resort to computational method in \Cref{subsubsec:PolicyIteration} to find the SPNE strategies. However, if players do have perfect observations, the backward enumeration approach can be done in polynomial time, which becomes a better choice than the policy iteration type of approach since it can guarantee the finding of the SPNE strategies and support online decision making.
\end{remark}

\section{NUMERICAL STUDY}\label{sec:NumericalStudy}
It is instructive to present some numerical studies of the results in \Cref{Sec:TheoreticalResults}. Our focus will be given to the observation and the jamming strategies. A scalar case will suffice to illustrate the interesting nature of the observation and the jamming sequences. A scalar case will even be more illustrative in some way. Our results and computational approaches can be effortlessly applied to higher dimensions. Suppose that
$$
\begin{aligned}
x_{n+1} &= ax_n + u^d_n + u_n^a + w_n,\\
\tilde{y}_n &= x_n + v_n,\\
y_n &= i^d_n(1-i^a_n) \tilde{y}_n,\\
\Sigma_s &= 4; \Sigma_0=1, \Sigma_o = \sigma,\\
Q_n &=1, R_n^d = 1, O_n^d = o^d, O_n^d = o^a\ \ \ 0\leq n\leq N-1\\
Q_N &= 8;, R_{N-1} = 10,\\
R_n^a & = r^a,\ \ \ 0\leq n \leq N-2
\end{aligned}
$$
We choose a total running time of $N=30$. The unassigned parameters include system matrix parameter $a$, the cost of physical attackers $r^a$, the observation noise variance $\sigma$, and the observation $o^d$ and jamming cost $o^a$, which are subject to change in the experiment. The computation follows the policy iteration-based algorithm.

\begin{figure}
    \centering
    \includegraphics[width=0.8\columnwidth]{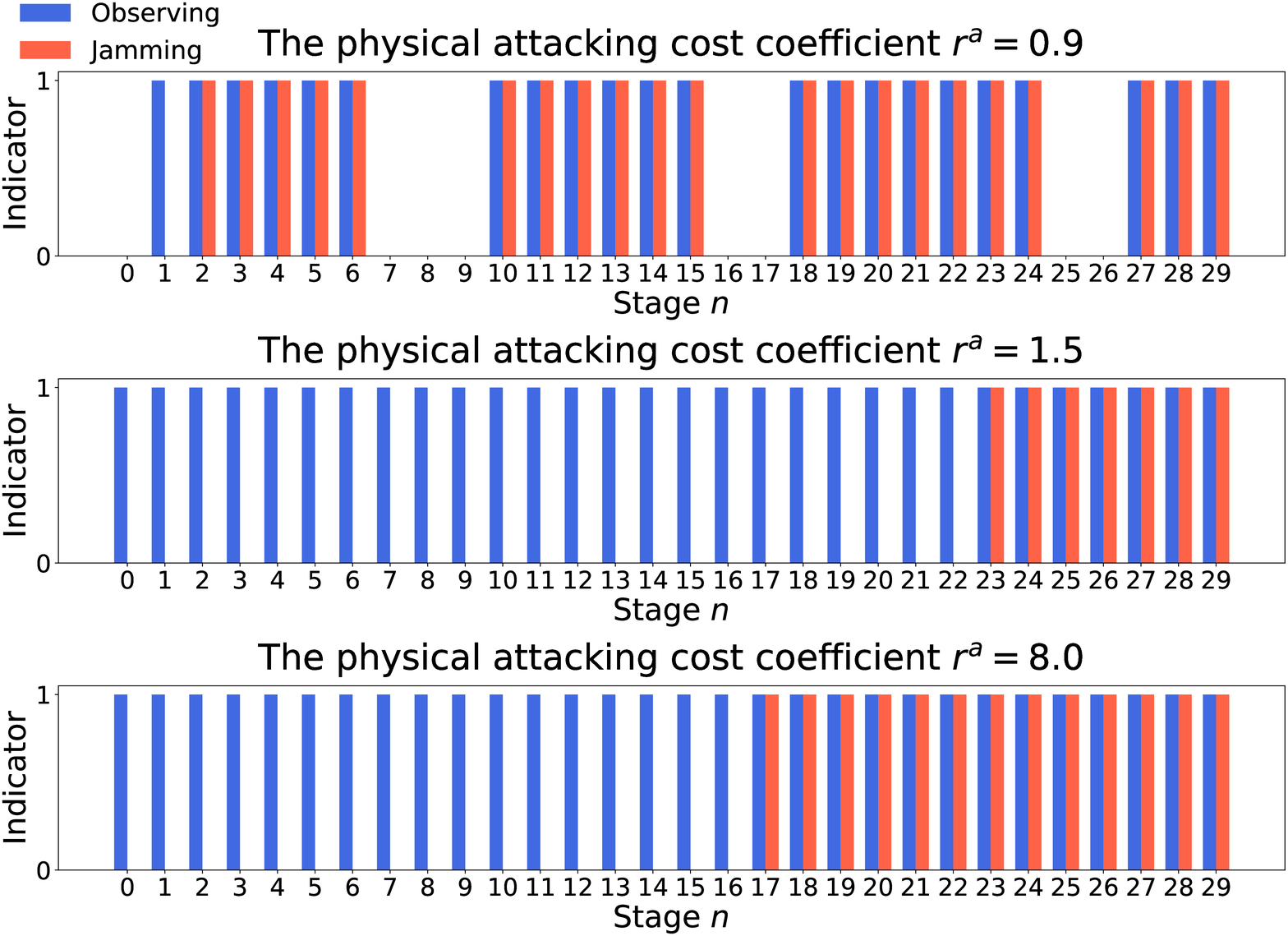}
    \caption{The observation decisions by the defender and the jamming decisions by the attacker over $30$ stages. Here, $a = 0.9$, $o_d =0$, $o_a = 15$, and  $\sigma =1$ .}
    \label{fig:ChangingAttackingCosts}
\end{figure}

\Cref{fig:ChangingAttackingCosts} shows the observation decision sequences and the jamming sequences for various costs of physical attacks. Fixed values of $a = 0.9$, $o_d =0$, $o_a = 15$, and  $\sigma =1$ are used. They give a slight stable system with zero cost of observation, cost of jamming being $15$, and variance of observation noise being $1$. The experiments are conducted for $r^a = 0.9, 1.5, \textrm{and }8.0$. The cost of attacking the physical plant has a strong impact on both the observation and the jamming decisions. When the cost of physical attacks is low, the defender does not observe at some stages even when there is no cost of observation. This is because observations at some stages will bring additional information that can be leveraged by the attacker, who is powerful in the physical side (i.e.,low cost of attacking), to compute her/his attacks on the physical plant. As the cost of physical attacks increases, undaunted by the additional information to the attacker, the defender observes at each stage when $r^a = 1.5$. As the cost of physical further increases to $r^a = 8.0$, the attacker enjoys less benefit from additional observations since her/his physical attacks are constrained by high costs. Hence, the attacker tends to jam more to prevent the defender from receiving observations.

\Cref{fig:ChangingSystemMatrix} shows the observation decision sequences and the jamming sequences under different system parameters. Fixed values of $r_a = 1.5$, $o_d =0$, $o_a = 15$, and  $\sigma =1$ are used. Sequences are shown for $a=0.5$, a highly stable system, $a=0.9$, a slightly stable system, and $a=1.1$, an unstable system. In all three cases, the defender chooses to observe at every stage because the attacks on the physical plant are limited by the high cost of attacking. Additional information will benefit the defender more. A highly stable system, say $a =0.5$, usually produces a very low-performance degradation when an observation is missing. Hence, intimidated by the cost of jamming, the attacker has no incentive to jam any of the observations. A slightly stable system can be affected by the attacker through jamming even if the cost of physical attacks is high. The attacker tends to jam economically. That is to jam near the end of the game to induce a considerable loss to the defender because $Q_N=8$ is much higher than $Q_n=1$ for $n\leq N-1$. When the defender deals with an unstable system, the attacker simply jams every stage so that the defender cannot stabilize the system due to no received observation. The defender could have chosen not to observe because the observation will be jammed anyway. But this is a zero-sum game, so the defender can at least gain a little from the attacker's cost induced by jamming.

\begin{figure}
    \centering
    \includegraphics[width=0.8\columnwidth]{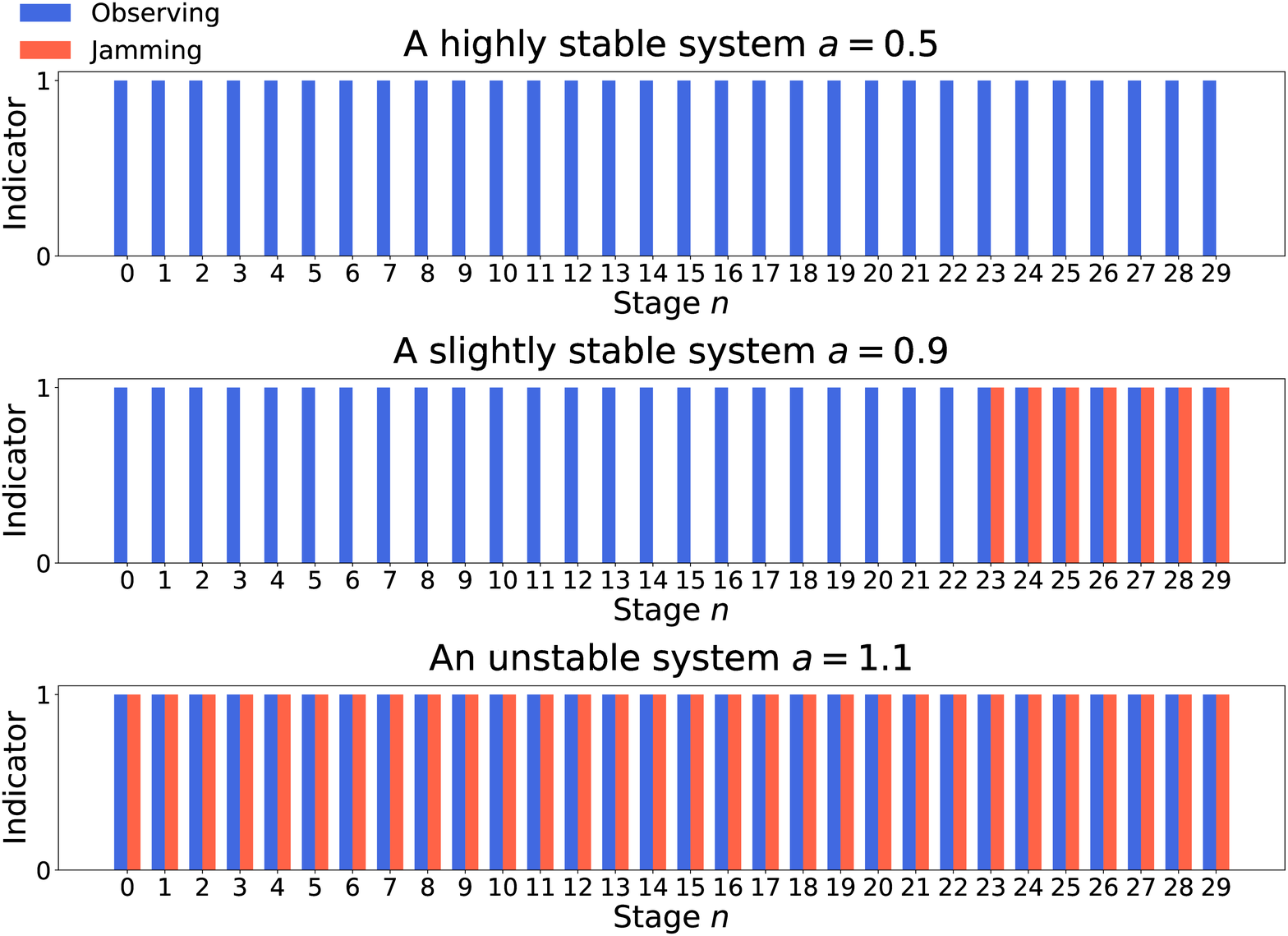}
    \caption{The Observation decisions by the defender and the jamming decisions by the attacker over $30$ stages. Here, $r_a = 1.5$, $o_d =0$, $o_a = 15$, and  $\sigma =1$ .}
    \label{fig:ChangingSystemMatrix}
\end{figure}

\Cref{fig:ChangingObvservationNoise} shows the number of observations and jammings as a function of observation noise variance $\sigma$. Fixed values $a = 1.1$, $o_d =25$, $o_a = 40$, and  $r_a =20$ are used. The cost of physical attacks is set to be high which means the attacker has limited capability at the physical side. When the cost of jamming is also high $o^a = 6000$, the game is degenerated to resemble an optimal control problem with observations that are costly and controlled. The curve regarding the number of observations shows some unusual results. As the observation noise variance $\sigma$ increases, the observation will be considered to be less valuable intuitively since it will contain less useful information about the state of the system. Grounded on this argument, one would expect the number of observations goes down monotonically to zero as $\sigma$ goes to infinity. Economically, it means that we should never pay for worthless information. However, the first blue curve of \Cref{fig:ChangingObvservationNoise} indicates that when the observation noise variance grows from a small to a moderately large value, it is better to actually increase the number of observations. This means when the information content of each individual observation is degraded slightly, it is better to pay the cost of making extra observations in order to make a better estimate. When the cost of jamming is lower, say when $o^a =40$, the defender observes more frequently when the observation noise variance is low because the defender knows that the attacker will be jamming (if the attacker does not jam, the defender will receive high-quality information to stabilize the system while the attacker can do little with the information). Doing so will render the attacker suffer more cost of jamming due to the fact that $O^a_n > O^d_n$. As the observation noise variance increases, the information becomes less useful. Only at some stages, the attacker has incentives to observe but most of these observations will be jammed since additional information is favored less by the attacker than the defender when the cost of physical attacks is high.

\begin{figure}
    \centering
    \includegraphics[width=0.8\columnwidth]{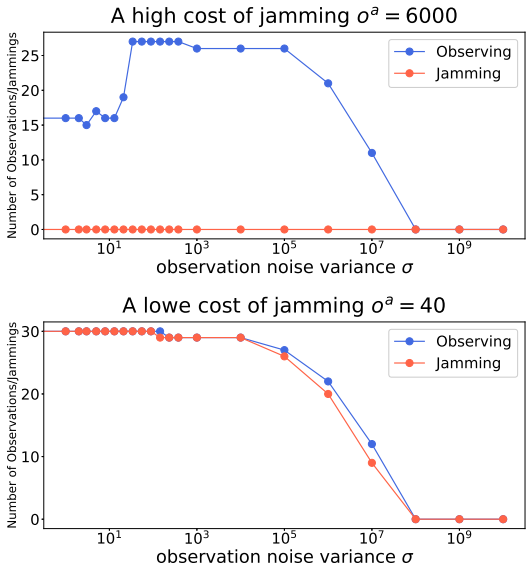}
    \caption{The number of observations and jamming attacks in $30$ stages. Here, $a = 1.1$, $o_d =25$, $o_a = 40$, and  $r_a =20$.}
    \label{fig:ChangingObvservationNoise}
\end{figure}
\section{Conclusions}

In this work, we have established a cross-layer multi-stage framework to facilitate the study of CPSs under coordinated cross-layer attacks. We have demonstrated that the framework is generic and can be specified to several classic attack models. The framework has been captured by a zero-sum linear quadratic Gaussian game with controlled observation. We have built solid theoretical underpinnings for this framework which can be used to analyze a wide variety of attacking settings. The theoretical results have shown that control performance depends on the observation and jamming strategies, which affects the quality of state estimation. Hence, the observation and jamming decisions can be carried out through dynamic equations that evolve as the estimation error variance propagates. Beyond that, the capability of altering the physical process will affect the jamming and observation decisions.

Future works will focus on the study of mixed strategies of observation and jamming, investigating the continuous-time scenario, and infinite-horizon problems.

\addtolength{\textheight}{-12cm}   









\bibliography{references}
\bibliographystyle{IEEEtran}

\end{document}